\documentclass[11pt]{article}
\usepackage[T1]{fontenc}
\usepackage{graphicx}
\usepackage{subcaption}
\usepackage{todonotes}
\usepackage{mathtools}
\usepackage{soul,comment,cite}
\usepackage{amsthm}
\usepackage{amsmath}
\usepackage{amssymb}
\usepackage{systeme}

\usepackage{gensymb}
\usepackage[ruled,vlined]{algorithm2e}
\SetKwComment{Comment}{/* }{ */}
\SetKwInOut{Input}{input}\SetKwInOut{Output}{output}
\newtheorem{theorem}{Theorem}
\newtheorem{corollary}{Corollary}
\newtheorem{lemma}{Lemma}
\newtheorem{observation}{Observation}

\newtheorem{definition}{Definition}

\graphicspath{{./figures/}}
\usepackage[absolute]{textpos}

\usepackage{fullpage}
\usepackage[top=0.8in, bottom=0.9in, left=0.9in, right=0.9in]{geometry}

\usepackage[pdftex, plainpages = false, pdfpagelabels, 
                 bookmarks=false,
                 bookmarksopen = true,
                 bookmarksnumbered = true,
                 breaklinks = true,
                 linktocpage,
                 pagebackref,
                 colorlinks = true,  
                 linkcolor = blue,
                 urlcolor  = cyan,
                 citecolor = red,
                 anchorcolor = green,
                 hyperindex = true,
                 hyperfigures
                 ]{hyperref}

\def\calD{\mathcal{D}}

\def\calH{\mathcal{H}}

\def\wvd{\mathcal{V}\!\mathcal{D}}
\def\calA{\mathcal{A}}
\def\calL{\mathcal{L}}
\def\fdvd{\mathcal{F}\!\mathcal{V}\!\mathcal{D}}

\title{Computing Dominating Sets in Disk Graphs with Centers in Convex Position\thanks{A preliminary version of this paper will appear in {\em Proceedings of the 17th Latin American Theoretical Informatics Symposium (LATIN 2026)}.}}

\author{
Anastasiia Tkachenko\thanks{Kahlert School of Computing,
University of Utah, Salt Lake City, UT 84112, USA. {\tt anastasiia.tkachenko@utah.edu}}
\and
Haitao Wang\thanks{Kahlert School of Computing,
University of Utah, Salt Lake City, UT 84112, USA. {\tt haitao.wang@utah.edu}}
}

\date{}

\begin{document}

\maketitle

\vspace{-0.2in}

\begin{abstract}
Given a set $P$ of $n$ points in the plane and a collection of disks centered at these points, the disk graph $G(P)$ has vertex set $P$, with an edge between two vertices if their corresponding disks intersect.
We study the dominating set problem in $G(P)$ under the special case where the points of $P$ are in convex position. The problem is NP-hard in general disk graphs. Under the convex position assumption, however, we present the first polynomial-time algorithm for the problem. Specifically, we design an $O(k^2 n \log^2 n)$-time algorithm, where $k$ denotes the size of a minimum dominating set. For the weighted version, in which each disk has an associated weight and the goal is to compute a dominating set of minimum total weight, we obtain an $O(n^5 \log^2 n)$-time algorithm.
\end{abstract}


\section{Introduction}
\label{sec:intro}

Let $P = \{p_1, \dots, p_n\}$ be a set of $n$ points in the plane, where each point $p_i \in P$ is assigned a radius $r_{p_i}\geq 0$. Let $D_{p_i}$ denote the disk centered at $p_i$ with radius $r_{p_i}$. 
The \emph{disk graph} $G(P)$ has $P$ as its vertex set, with an edge between $p_i$ and $p_j$ if the two disks $D_{p_i}$ and $D_{p_j}$ intersect, i.e., $|p_ip_j|\leq r_{p_i}+r_{p_j}$, where $|p_ip_j|$ is the Euclidean distance of $p_i$ and $p_j$. Disk graphs, including the {\em unit-disk case}, where all disks have the same radius, arise naturally in numerous application domains, such as wireless sensor networks (where devices with heterogeneous transmission ranges communicate when within range)~\cite{ref:PerkinsAd99,ref:BalisterCo05,ref:KammerlanderC13,ref:NguyenDi10}, surveillance networks~\cite{ref:TruongMo12}, optimal facility location~\cite{ref:ClarkUn90,ref:WangA88}, etc. 

Many classical problems on disk graphs are still NP-hard~\cite{ref:ClarkUn90} and approximation algorithms have been developed, such as vertex cover~\cite{ref:ErlebachPo05,ref:LokshtanovA23,ref:LeeuwenBe06}, independent set~\cite{ref:ErlebachPo05},  feedback vertex set~\cite{ref:BandyapadhyaySu22}, dominating set~\cite{ref:GibsonAl10}, etc. One exception is that finding a maximum clique in a unit-disk graph can be solved in polynomial time~\cite{ref:ClarkUn90,ref:EppsteinGr09,ref:EspenantFi23}. But it has been a long standing open problem whether the clique problem on general disk graphs can be solved in polynomial time; see \cite{ref:BonamyEP21,ref:KeilTh25} for some recent progress on this problem.

We are interested in the dominating set problem in disk graphs. 
A \textit{dominating set} of $G(P)$ is a subset $S \subseteq P$ such that every vertex in $G(P)$ is either in $S$ or adjacent to a vertex in $S$. The \textit{dominating set problem} is to find a dominating set of minimum cardinality. In the {\em weighted case}, each point of $P$ has a weight, and the goal is to find a dominating set minimizing the total weight. 
The problem is NP-hard even in unit-disk graphs~\cite{ref:ClarkUn90}. Approximation algorithms have been proposed, e.g., ~\cite{ref:GibsonAl10, ref:HernandezAp22, ref:MustafaIm10, ref:DeGe23}.

\subsection{Convex position setting and previous work}
In light of the above hardness results, exploring structured settings that may allow efficient algorithms is of particular interest. 
In this paper, we consider the dominating set problem in disk graphs in a {\em convex position} setting where all points of $P$ are in convex position, that is, every point of $P$ appears as a vertex of the convex hull of $P$. This setting models deployments along perimeters (e.g., fences, shorelines, or rings of sensors) and, importantly, reveals strong structural relationships. 

In this convex position setting, Tkachenko and Wang~\cite{ref:TkachenkoDo25} studied the problem for the unit-disk case. Specifically, they gave an $O(n^3\log^2 n)$ time algorithm for the weighted dominating set problem, and an $O(kn\log n)$ time algorithm for the unweighted case, where $k$ is the size of a minimum dominating set.

Other geometric problems in convex position have also been studied. For the independent set problem, which is to find a maximum subset of vertices of $G(P)$ so that no two vertices have an edge, an $O(n^{7/2})$ time algorithm was given in \cite{ref:TkachenkoDo25} in the convex-position setting. The classical $k$-center problem among a set of points in the plane is to find $k$ congruent disks of smallest radius to cover all points. The problem is NP-hard. Choi, Lee, and Ahn~\cite{ref:ChoiEf23} studied the problem in the convex position setting and proposed an $O(n^3\log n)$ time algorithm. For the discrete version of the problem where the centers of the disks are required to be in the given point set, Tkachenko and Wang solved the problem in $O(n^2\log^2 n)$ time~\cite{ref:TkachenkoDo25}. Given a set of points in the plane and a number $k$, the dispersion problem is find $k$ points so that their minimum pairwise distance is maximized. The problem is also NP-hard~\cite{ref:WangA88}. Singireddy, Basappa, and Mitchell~\cite{ref:SingireddyAl23} studied the problem in the convex position setting and gave an $O(n^4k^2)$ time algorithm. An improved algorithm of $O(n^{7/2}\log n)$ time was later derived in \cite{ref:TkachenkoDo25}.

Even problems that are already polynomial-time solvable in the general case also attracted attention in the convex position setting, e.g., the classic linear-time algorithm of Aggarwal, Guibas, Saxe, and Shor~\cite{ref:AggarwalA89} to construct Voronoi diagrams for a set of 2D points in convex position. Refer to~\cite{ref:LingasOn86, ref:RichardsRe90, ref:ChazelleIm93, ref:TkachenkoCo25, ref:BiniazA16,ref:CagiMa18} for additional work in convex position setting.



\subsection{Our result}

We study the dominating set problem in disk graphs under the convex position setting. While polynomial-time algorithms are known for the unit-disk case~\cite{ref:TkachenkoDo25}, allowing disks of varying radii introduces significant challenges. For instance, although a polynomial-time algorithm for computing a maximum clique in unit-disk graphs was proposed more than 30 years ago~\cite{ref:ClarkUn90}, it remains unknown whether the same problem can be solved in polynomial time for general disk graphs.

Nevertheless, we attempt to extend the techniques of \cite{ref:TkachenkoDo25} to the general disk graph setting. This proves far from trivial, as many properties of the unit-disk case no longer hold. Even so, we uncover new structural observations and develop novel algorithmic techniques. Consequently, we establish that the dominating set problem in disk graphs under the convex position setting can be solved in polynomial time.

Specifically, we present an $O(n^5 \log^2 n)$-time algorithm for the weighted case. Furthermore, given a size bound $k$, we can compute a minimum-weight dominating set of size at most $k$ (if one exists) in $O(k^2 n^3 \log^2 n)$ time. Our approach is a dynamic programming algorithm that leverages a key structural property we prove, which we call the line-separable property. This property enables us to decompose the original problem into smaller subproblems amenable to dynamic programming.

For the unweighted case, we design a more efficient algorithm by employing a greedy strategy, achieving a running time of $O(k^2 n \log^2 n)$, where $k$ denotes the size of a minimum dominating set. In particular, when $k = O(1)$, our algorithm runs in $O(n \log^2 n)$ time. 


While problems in unit-disk graphs under the convex position setting have been studied previously, to the best of our knowledge, little work has been done on general disk graphs in convex position. We therefore hope that our result will stimulate further research in this direction.


\paragraph{Outline.}
The rest of the paper is organized as follows. After introducing notation in Section~\ref{sec:pre},  
Section~\ref{sec:structure} presents some structural properties that both our weighted and unweighted algorithms reply on. 
The algorithms for the weighted case and the unweighted case are given in Sections~\ref{sec:domwgt} and~\ref{sec:domunwgt}, respectively. 

\section{Notation}
\label{sec:pre}
We introduce some notations that will be used throughout the paper, in addition to those already defined in Section~\ref{sec:intro}, e.g., $P$, $n$, $G(P)$, $D_{p_i}$. 

For any two points $p$ and $q$ in the plane, we use $|pq|$ to denote their (Euclidean) distance. We also use $\overline{pq}$ to denote the line segment connecting them. 

Let $\calH(P)$ denote the convex hull of $P$. We assume that the points of $P$ are in convex position. Then $P$ can be considered a cyclic sequence around $\calH(P)$. Specifically, let $P = \langle p_1, p_2, \ldots, p_n \rangle$ represent a cyclic list of the points ordered counterclockwise along $\calH(P)$. We use a {\em sublist} to refer to a contiguous subsequence of $P$. Multiple sublists are said to be \textit{consecutive} if their concatenation is also a sublist. 
For any two points $p_i$ and $p_j$ in $P$, we define $P[i,j]$ as the sublist of $P$ from $p_i$ counterclockwise to $p_j$, inclusive, i.e., if $i \leq j$, then $P[i,j] = \langle p_i, p_{i+1}, \ldots, p_j \rangle$; otherwise, $P[i,j] = \langle p_i, p_{i+1}, \ldots, p_n, p_1, \ldots, p_j \rangle$. Note that if $i=j$, then $P[i,j]=\langle p_i\rangle$. 
We also denote by $P(i,j]$ the sublist $P[i,j]$ excluding $p_i$, 
and similarly for other variations, e.g., $P[i,j)$ and $P(i,j)$.

For any two points $p_i,p_j\in P$, we define $|D_{p_i}D_{p_j}|=|p_ip_j|-( r_{p_i}+r_{p_j})$, and we call it the {\em disk distance} between $D_{p_i}$ and $D_{p_j}$. If the two disks intersect, then $|D_{p_i}D_{p_j}|\leq 0$; otherwise, $|D_{p_i}D_{p_j}|> 0$. 

For any subset $P'\subseteq P$, let $\calD(P')=\{D_{p_i}\ |\ p_i\in P\}$. Sometimes we also say that disks of $\calD(P)$ are {\em input disks}. We emphasize that while the centers of the input disks, i.e., the points of $P$, are in convex position, the disks themselves may not be, i.e., it is possible that an input disk $D_{p_i}$ is in the interior of the convex hull of all disks of $\calD(P)$.






\section{Structural properties}
\label{sec:structure}
 
We begin by examining the structural properties of dominating sets in $G(P)$. We introduce these properties for the weighted case, which are applicable to the unweighted case too.

For a sublist $\alpha$ of $P$, we say that a point $p_i\in P$ {\em dominates} $\alpha$ if the disk $D_{p_i}$ intersects $D_{p}$ for all points $p\in \alpha$. For two points $p_i,p_j\in P$, if $D_{p_i}$ intersects $D_{p_j}$, then we also say that $p_i$ {\em dominates} $p_j$ (and similarly, $p_j$ dominates $p_i$). 

Suppose $S \subseteq P$ is a dominating set of $G(P)$.  Let $\calA$ be a partition of $P$ into (nonempty) disjoint sublists such that for every sublist $\alpha \in \calA$, there exists a point in $S$ that dominates $\alpha$.
An \textit{assignment} is a mapping $\phi: \calA \to S$ that assigns every sublist $\alpha \in \calA$ to exactly one point $p_i \in S$ such that $p_i$ dominates $\alpha$. 
For each $p_i \in S$, we define the \textit{group} of $p_i$, denoted by $\calA_{p_i}$, as the collection of all sublists $\alpha \in \calA$ that are assigned to $p_i$ under $\phi$. 
Depending on the context, $\calA_{p_i}$ may also refer to the set of points of $P$ in all its sublists. For example, if a point $p_j$ is in a sublist of $\calA_{p_i}$, then we may write ``$p_j\in \calA_{p_i}$''.
By definition, the groups of all points of $S$ are pairwise disjoint, and they together form a partition of $\calA$ and also form a partition of $P$.

An assignment $\phi: \calA \to S$ is \textit{line-separable} if for every two points $p_i,p_j\in S$, there exists a line $\ell$ such that the points of $\calA_{p_i}$ lie entirely on one side of $\ell$, while the points of $\calA_{p_j}$ lie on the other side.

The following lemma proves a {\em line-separable property}. As will be seen, the property is crucial for our dominating set algorithms.


\begin{lemma}
\label{lem:linesep}
Suppose $S$ is an optimal dominating set of $G(P)$. Then there exists a partition $\calA$ of $P$ and a line-separable assignment 
$\phi : \calA\rightarrow S$ such that (1) for every point $p_i\in S$, $p_i \in \calA_{p_i}$, i.e., the group $\calA_{p_i}$ contains $p_i$ itself, and (2) any two adjacent sublists of $\calA$ are assigned to different points of $S$.
\end{lemma}

\begin{proof}

We first describe how to construct the partition $\calA$ and the assignment $\phi$, and then argue that $\phi$ is line-separable. 

\paragraph{Constructing $\calA$ and $\phi$.}
Let $\wvd(S)$ be the additively-weighted Voronoi diagram of the points of $S$ with the weight of each point $p\in S$ as $-r_p$, i.e., the weighted distance $d_p(q)$ between any point $q$ and $p$ is defined as $|pq|-r_p$. Let $R(p)$ denote the Voronoi cell of $p$ in $\wvd(S)$. Note that $R(p)$ is star-shaped and contains $p$~\cite{ref:SharirIn85, ref:FortuneA87}.

Since $S$ is an optimal dominating set, $S$ does not have two points $p_i$ and $p_j$ such that $D_{p_i}\subseteq D_{p_j}$ (since otherwise $S\setminus\{p_i\}$ is still a dominating set of $G(P)$, contradicting the optimality of $S$). Because of this, every point $p$ of $S$ has a non-empty Voronoi cell in $\wvd(S)$ that contains $p$~\cite{ref:SharirIn85, ref:FortuneA87}, i.e., $p\in R(p)$.

For each point $p_i\in S$, for every point $p \in P$ that is contained in $R(p_i)$, we assign $p$ to the group $\calA_{p_i}$. If $p$ lies on the boundary of two Voronoi cells, then we assign $p$ to an arbitrary one of them. In this way, each point $p\in P$ is assigned to exactly one point of $S$. In addition, since $p_i\in R(p_i)$, $p_i$ is assigned to its own group $\calA_{p_i}$ (this guarantees the first condition in the lemma statement). 

For each point $p_i\in S$, we make each contiguous maximal sequence of points of $\calA_{p_i}$ form a a sublist (this guarantees the second condition in the lemma statement). Let $\calA$ be the union of $\calA_{p_i}$ for all points $p_i\in S$. The assignment $\phi$ simply follows from the above definition of groups $\calA_{p_i}$. Note that $\calA$ satisfies both conditions in the lemma statement.

To show that $\calA$ is a desired partition of $P$ for $S$, it suffices to show,  for any point $p_i\in S$, $D_{p_i}$ intersects $D_{p}$ for all points $p\in \calA_{p_i}$. Indeed, consider a point $p\in \calA_{p_i}$. Since $S$ is a dominating set, $D_{p}$ must intersect $D_{p_j}$ for some point $p_j\in S$. Thus, $d_{p_j}(p)\leq r_{p}$. Since $p\in \calA_{p_i}$, by our construction of $\calA_{p_i}$, $p\in R(p_i)$. Hence, we have $d_{p_i}(p)\leq d_{p_j}(p)\leq r_p$ and thus $D_p$ and $D_{p_i}$ intersect.

\paragraph{Proving the line-separable property.}  
It remains to prove that the assignment $\phi$ is line-separable. 
Suppose to the contrary that there exist two groups $\calA_{p_i}$ and $\calA_{p_j}$ (with $p_i,p_j \in S$) that are not line-separable.  
Then, since $p_i\in \calA_{p_i}$ and $p_j\in \calA_{p_j}$, there must be a point $p\in \calA_{p_i}$ and a point $p'\in \calA_{p_j}$ such that $\overline{p_ip}$ and $\overline{p_jp'}$ cross each other. On the other hand, since $p\in \calA_{p_i}$, $p$ is in $R(p_i)$. As $R(p_i)$ is star-shaped with respect to $p_i$~\cite{ref:SharirIn85, ref:FortuneA87}, $\overline{p_ip}\subseteq R(p_i)$. Similarly, $\overline{p_jp'}\subseteq R(p_j)$. Since $R(p_i)$ and $R(p_j)$ are interior disjoint, $\overline{p_ip}$ cannot cross $\overline{p_jp'}$. We thus obtain contradiction. 
%
%
%
\end{proof}

Following the notation in Lemma~\ref{lem:linesep},
for each center $p_i\in S$, we call the sublist of $\calA_{p_i}$ containing $p_i$ the 
\textit{main sublist} of $p_i$, and all other sublists of $\calA_{p_i}$ are called \textit{secondary sublists}. 

\paragraph{Remark.}
In the unit-disk case, for any optimal dominating set $S$, there always exists a partition $\calA$ and an assignment of $\phi$ such that for each $p_i\in S$, the group $\calA_{p_i}$ contains at most two sublists~\cite{ref:TkachenkoDo25}. This property leads to a more efficient algorithm for the unit-disk case in \cite{ref:TkachenkoDo25}. Unfortunately, the property no longer holds anymore for our general disk problem. Refer to Figure~\ref{fig:OK} for an example, in which there is a large disk $D_{p_i}$ while other disks are disjoint and arranged in such a way that they alternatively intersect and avoid $D_{p_i}$. Assuming that all points have the same weight, $p_i$ will be included in any optimal dominating set, and $p_i$ will be assigned $\Omega(n)$ sublists. Comparing with the unit-disk case, this ``unbounded'' group size makes it substantially more challenging to design efficient algorithms. Nonetheless, we demonstrate below that by relying solely on the line-separable property from Lemma~\ref{lem:linesep}, we can still derive a polynomial-time algorithm.
\medskip

\begin{figure}[t]
\begin{minipage}[t]{\textwidth}
\begin{center}
\includegraphics[trim={0 3.5cm 0 1cm},clip,height=1.7in]{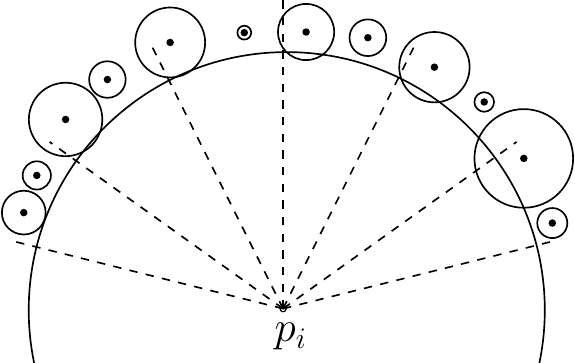}
\caption{\footnotesize Illustrating an example in which $|\calA_{p_i}|=\Omega(n)$.}
\label{fig:OK}
\end{center}
\end{minipage}
\vspace{-0.1in}
\end{figure}

We further have the following lemma, which will be instrumental in our algorithm design. 

\begin{lemma}\label{lem:onesizegroup}
Let $S$ be an optimal dominating set of $G(P)$, and let $\phi : \calA \rightarrow S$ be the assignment given by Lemma~\ref{lem:linesep}. 
Then, $S$ has a point $p_i$ whose group $\calA_{p_i}$ only has a single sublist (which is the main sublist of $p_i$).
\end{lemma}
\begin{proof}
For each sublist $\alpha$ of $\calA$, we pick a point of $\alpha$ as its {\em representative point}. For each point $p_i\in S$, we order the sublists of $\calA_{p_i}$ along $\calH(P)$ and for every two adjacent sublists in that order we connect their representative points by a line segment, called a {\em diagonal} of $P$. Let $E$ be the set of all diagonals defined above for all $p_i\in S$. 

As the assignment $\phi$ is line-separable, no two diagonals of $E$ cross each other. Hence, there must exist a diagonal $d\in E$ that divides $\calH(P)$ into two subpolygons one of which does not contain any other diagonal; let $R$ denote such a subpolygon. Note that $d$ connects representative points of two sublist from a group $\calA_{p_i}$ of some point $p_i$. Since no two sublists of $\calA_{p_i}$ are consecutive along $\calH(P)$ by the second property in Lemma~\ref{lem:linesep}, $R$ must contain a sublist $\alpha$ belonging to a group $\calA_{p_j}$ with $p_j\neq p_i$. As $R$ does not contain any diagonal other than $d$, the representative point of $\alpha$ is not connected by any diagonal of $E$. This means that $\alpha$ is the only sublist of $\calA_{p_j}$. The lemma thus follows. 
\end{proof}

\section{The weighted dominating set problem}
\label{sec:domwgt} 

In this section, we present our algorithm for computing a minimum-weight dominating set in the disk graph $G(P)$. 

For each point $p_i \in P$, let $w_i$ denote its weight. We assume $w_i > 0$ since otherwise $p_i$ can always be included in the solution. For any subset $P' \subseteq P$, define $w(P')= \sum_{p_i \in P'} w_i$.

We will focus on the following \textit{bounded-size problem}: Given an integer $k$, find a dominating set $S \subseteq P$ of minimum total weight in $G(P)$ subject to $|S|\leq k$. Solving this problem with $k = n$ yields a minimum-weight dominating set for $G(P)$.
Let $W^*$ denote the total weight of a minimum-weight dominating set of size at most $k$. 

The rest of this section is organized as follows. Section~\ref{sec:overview} first introduces a new concept {\em rank-$t$ centers}, which is crucial to our algorithm, and also provides an overview about our algorithm. Section~\ref{sec:description} gives the details of the algorithm while the correctness is proved in Section~\ref{sec:correct}. The algorithm implementation and time analysis are finally discussed in Section~\ref{sec:time}.

\subsection{Rank-$\boldsymbol{t}$ centers and algorithm overview}
\label{sec:overview}


Let $S$ be an optimal dominating set of $G(P)$ with $|S|\leq k$, and let $\phi : \calA \rightarrow S$ be the assignment given by Lemma~\ref{lem:linesep} (note that the lemma was originally for an optimal dominating set without bounded-size constraint, but the proof works for bounded-size optimal dominating set too).

Let $\alpha_1, \alpha_2, \ldots, \alpha_m$ be the sublists of $\calA$ following the order of $P$, i.e., counterclockwise around $\calH(P)$. For any two sublists $\alpha_g$ and $\alpha_h$, we define $\calA[g,h]$ as the set of the sublists from $\alpha_g$ counterclockwise to $\alpha_h$ inclusive. Depending on the context, $\calA[g,h]$ may also refer to the union of all these sublists, which form a sublist of $P$. 
We define $S[g,h]$ as the subset of points $p\in S$ whose groups $\calA_p$ contain a sublist of $\calA[g,h]$. 
Note that if $g=h$, then $\calA[g,h]$ refers to $\{\alpha_g\}$. 

\begin{definition}\label{def:tcenter}
We say that a point $p_i\in S$ is a {\em rank-$t$ center} for $\calA[g,h]$ if the following conditions are satisfied: (1) $|S[g,h]|\leq t$; (2) at least one of $\alpha_g$ and $\alpha_h$ is in $\calA_{p_i}$; (3) the main sublist of $p_i$ is in $\calA[g,h]$; (4) if $t>1$, then for each point $p_j\in S[g,h]\setminus\{p_i\}$, $\calA_{p_j}\subseteq \calA[g,h]$, i.e., all sublists in the group $\calA_{p_j}$ are in $\calA[g,h]$; otherwise, $\calA_{p_i}\subseteq \calA[g,h]$.     
\end{definition}
Note that this definition does not require $\calA_{p_i}\subseteq \calA[g,h]$ if $t>1$, and therefore it is possible that some secondary sublists of $\calA_{p_i}$ are not in $\calA[g,h]$. We also say that $\calA[g,h]$ in the above definition is a {\em rank-$t$ sublist} of $p_i$. 
In addition, in the second condition of the definition, if exactly one of $\alpha_g$ and $\alpha_h$ is in $\calA_{p_i}$, then 
we call $p_i$ a rank-$t$ {\em open} center, and if both $\alpha_g$ and $\alpha_h$ are in $\calA_{p_i}$, then 
$p$ a rank-$t$ {\em closed} center. 

To see why the concept rank-$t$ center is useful, observe that every point $p_i\in S$ is a rank-$k$ center for some sublist $\calA[g,h]$ that is $P$. Indeed, let $\alpha_g\in \calA$ be the main sublist of $p_i$ for some index $g$. Then, since $|S|\leq k$, $p_i$ is a rank-$k$ center for $\calA[g-1,g]$ or $\calA[g,g+1]$, where the operations ``$g-1$'' and ``$g+1$'' are module $m$ (we follow this convention in the rest of this section). Note that both $\calA[g-1,g]$ and $\calA[g,g+1]$ are $P$ (and both $S[g-1,g]$ and $S[g,g+1]$ are $S$). 

In addition, at least one point $p_i\in S$ is a rank-$1$ center of some sublist of $\calA$. Indeed, by Lemma~\ref{lem:onesizegroup}, there exists a center $p_i\in S$ whose group $\calA_{p_i}$ has only its main sublist, denoted by $\alpha_g$. Then, $p_i$ is a rank-$1$ center of the sublist $\calA[g,g]=\{\alpha_g\}$. 

\paragraph{Algorithm overview.}
Our algorithm is a dynamic program with $k$ iterations. In each $t$-th iteration, $1 \leq t \leq k$, we compute for each point $p_i\in P$ a set $\calL_t(i)$ of $O(n^2)$ sublists of $P$, with each sublist $L \in \calL_t(i)$ associated with a value $w'(L)$ and a subset $S_L \subseteq P$, such that the following {\em algorithm invariants} are maintained: (1) $w(S_L) \leq w'(L)$; (2) $S_L$ dominates $L$; (3) $p_i\in S_L$; (4) $|S_L|\leq t$. Define $\calL_t=\cup_{p_i\in P}\calL_t(i)$. 

To argue the correctness of the algorithm, we will show that for each point $p_i\in P$ and each rank-$t$ sublist $\calA[g,h]$ of $p_i$, $\calL_t(i)$ must contain a sublist $L$ with $\calA[g,h]\subseteq L$ and $w(S_L)\leq w'(L)\leq w(S[g,h])$. Because every point $p_i\in S$ is a rank-$k$ center for some sublist that is $P$, after the $k$-th iteration, $\calL_k$ must contain a sublist $L$ that is $P$ with $w(S_L)\leq w'(L) \leq w(S)=W^*$, implying that $S_L$ is an optimal dominating set (of size at most $k$). As such, after $k$ iterations, among all sublists in $\calL_k$ that are $P$, we find the one $L$ with minimum $W'(L)$ and return $S_L$ as the optimal dominating set.
Note that because at least one point $p_i\in S$ is a rank-$1$ center of some sublist of $\calA$, our algorithm starts with $t=1$. 


\medskip
The following notation will be used in the rest of this section. 
\begin{definition}
\label{def:ab}
For two points $p_i,p_j\in P$ ($p_i=p_j$ is possible), define $a_{i}^j$ as the index of the first point $p$ of $P$ counterclockwise from $p_j$ such that $|D_{p_i}D_p|> 0$, and $b_{i}^j$ the index of the first point $p$ of $P$ clockwise from $p_j$ such that $|D_{p_i}D_p|> 0$ (if $|D_{p_i}D_{p_j}|>0$, then $a_{i}^j=b_{i}^j=j$). If $|D_{p_i}D_p| \leq 0$ for all points $p\in P$, then let $a_{i}^j=b_{i}^j=0$.   
\end{definition}

\subsection{Algorithm description}
\label{sec:description}
We now describe the details of each $t$-th iteration of the algorithm with $1\leq t\leq k$. 

Initially, $t=1$, and our algorithm computes two indices $a_i^i$ and $b_i^i$ as defined in Definition~\ref{def:ab} for each $p_i\in P$. We will show in Lemma~\ref{lem:firstout} that this can be done in $O(n \log^2 n)$ time.
Then, for each $p_i\in P$, let $L=P(b_i^i, a_i^i)$, $S_L=\{p_i\}$, $w'(L)=w_i$, and $\calL_1(i)=\{L\}$. Obviously, all algorithm invariants hold for $L$. This completes the first iteration of the algorithm. 

Suppose that for each $t'$, $1 \leq t' \leq t - 1$, we have computed a collection $\calL_{t'}(i)$ of $O(n^2)$ sublists for each $p_i\in P$, with each sublist $L \in \calL_{t'}(i)$ associated with a value $w'(L)$ and a point set $S_L \subseteq P$, such that the algorithm invariants hold for $L$, i.e., $w(S_L) \leq w'(L)$, $S_L$ dominates $L$, $p_i\in S_L$, and $|S_L|\leq t'$.


The $t$-th iteration of the algorithm works as follows. For each point $p_i \in P$, we perform a counterclockwise/clockwise processing procedure for each point $p_j$, and also perform a bidirectional preprocessing procedure. The details are given below.

\subsubsection{Counterclockwise/clockwise processing procedures}
For each other point $p_j \in P$, we perform a \textit{counterclockwise processing procedure}, as follows. 
For each $t'$ with $1\leq t'\leq t-1$, and each point $p_z\in P[i,j]$, we do the following. Refer to Figure~\ref{fig:sublist10}.

We first perform a \textit{minimum-value enclosing sublist query} on $\calL_{t'}(i)$ to find the sublist $L_1\in \calL_{t'}(i)$ of minimum $w'(L_1)$ such that $P[i, z]\subseteq L_1$; we will discuss how to solve the query in Lemma~\ref{lem:enclosequery}. Let $p_{z_1}$ be the counterclockwise endpoint of $L_1$. Then we perform another minimum-value enclosing sublist query on $\calL_{t - t'}$ to find the sublist $L_2\in \calL_{t - t'}$ of minimum $w'(L_2)$ with $P[z_1+1, j]\subseteq L_2$. Let $p_{z_2}$ be the counterclockwise endpoint of $L_2$. Next, we compute the index $a_i^{z_2+1}$. 

\begin{figure}[h]
\begin{minipage}[h]{0.49\textwidth}
\begin{center}
\includegraphics[height=2.0in]{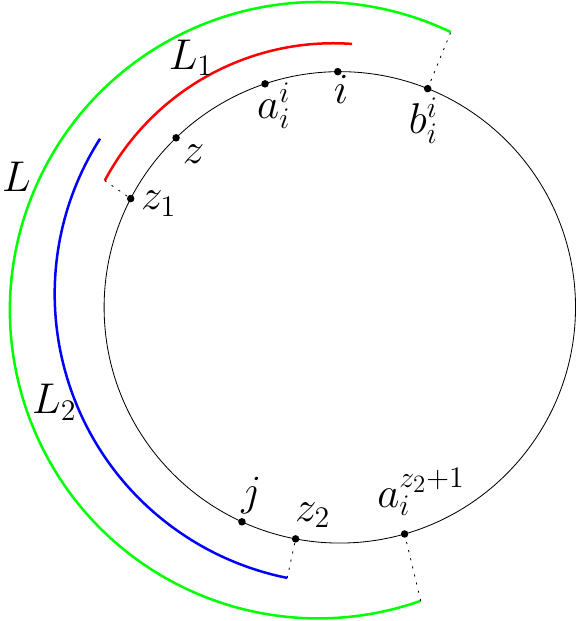}
\caption{\footnotesize Illustrating the relative positions of points of $P$ (only their indices are shown): the circle represent $\calH(P)$.}
\label{fig:sublist10}
\end{center}
\end{minipage}
\hspace{0.05in}
\begin{minipage}[h]{0.49\textwidth}
\begin{center}
\includegraphics[height=1.8in]{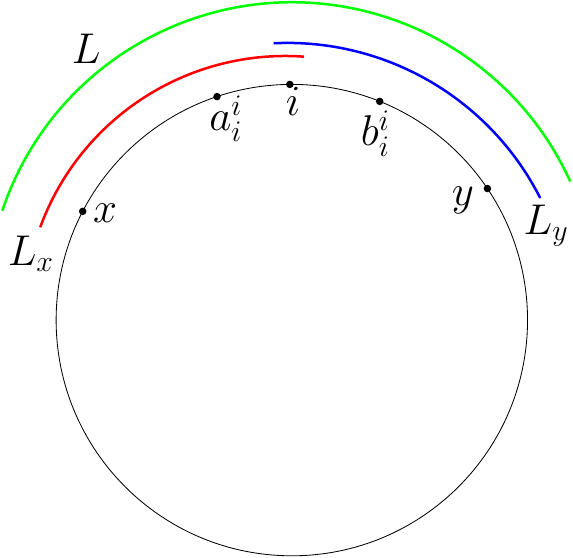}
\caption{\footnotesize Illustrating $L_x$, $L_y$, and $L$.}
\label{fig:sublist20}
\end{center}
\end{minipage}
\vspace{-0.1in}
\end{figure}

By construction, the union of the following four sublists are consecutive and thus form a sublist of $P$: $P(b_i^i, a_i^i)$, $L_1$, $L_2$ and $P(z_2, a_i^{z_2+1})$. Denote this combined sublist by $L(t', m)$, or $L$ for simplicity if $t'$ and $m$ are clear from the context. We set $S_L= S_{L_1} \cup S_{L_2}$ and $w'(L) = w'(L_1) + w'(L_2)$.
\begin{observation}
All algorithm invariants hold for $L$. 
\end{observation}
\begin{proof}
We argue that the algorithm invariants hold for $L$, i.e., $w(S_L) \leq w'(L)$, $L$ is dominated by $S_L$, $p_i\in S_L$, and $|S_L| \leq t$. Indeed, since $L_1\in \calL_{t'}(i)$, by the algorithm invariants for $t'$, $w(S_{L_1}) \leq w'(L_1)$, $L_1$ is dominated by $S_{L_1}$, $p_i \in S_{L_1}$, and $|S_{L_1}| \leq t'$. Similarly, since $L_2\in \calL_{t - t'}$, by the algorithm invariants for $t-t'$, $w(S_{L_2}) \leq w'(L_2)$, $L_2$ is dominated by $S_{L_2}$, and $|S_{L_2}| \leq t - t'$. In addition, by definition, points of $P(b_i^i, a_i^i)$ and $P(z_2, a_i^{z_2+1})$ are dominated by $p_i$. Combining all these, we obtain that all algorithm invariants hold for $L$.   
\end{proof}

Among all such sublists $L(t', m)$ with $1\leq t'\leq t-1$ and $p_z\in P[i,j]$, we keep the one with minimum $w'(L(t',m))$ and add it to $\calL_t(i)$. This finishes the counterclockwise processing procedure for $p_j$. 

Symmetrically, we also perform a \textit{clockwise processing procedure} for $p_j$, which adds at most one sublist to $\calL_t(i)$. 

Performing these two procedures for all $p_j\in P$ adds $O(n)$ sublists to $\calL_t(i)$. 




\subsubsection{Bidirectional processing procedures}
We also perform \textit{bidirectional processing procedures} for $p_i$. 
For each pair of points $(p_x,p_y)$ such that $p_y,p_i,p_x$ are in counterclockwise order along $\calH(P)$, we do the following for each $t'$ with $2 \leq t' \leq t - 1$. Refer to Figure~\ref{fig:sublist20}.

We perform a minimum-value enclosing sublist query on $\calL_{t'}(i)$ to find the sublist $L_x\in \calL_{t'}(i)$ of minimum  $w'(L_x)$ such that $P[i,x]\subseteq L_x$. Similarly, we perform a minimum-value enclosing sublist query on $\calL_{t+1-t'}(i)$ to find the sublist $L_y\in \calL_{t+1-t'}(i)$ of minimum $w'(L_y)$ such that $P[y,i]\subseteq L_y$. By definition, the union of $P(b_i^i,a_i^i)$, $L_x$, and $L_y$ are consecutive and thus form a sublist of $P$, denoted by $L(x,y,t')$ or simply $L$. We set $S_L = S_{L_x} \cup S_{L_y}$ and $w'(L)= w'(L_x) + w'(L_y) - w_i$. 
\begin{observation}
All algorithm invariants hold for $L$. 
\end{observation}
\begin{proof}
Indeed, by the algorithm invariants for $t'$ and $t-t'+1$, the sublists $L_x$ and $L_y$ are dominated by $S_{L_x}$ and $S_{L_y}$, respectively, $w(S_{L_x}) \leq w'(L_x)$, $w(S_{L_y}) \leq w'(L_y)$, $p_i\in S_{L_x}$, $p_i\in S_{L_y}$, $|S_{L_x}|\leq t'$, and $|S_{L_y}|\leq t-t'+1$.
Since both $S_{L_x}$ and $S_{L_y}$ include $p_i$ and $P(b_i^i,a_i^i)$ is dominated by $p_i$, $S_L$ has size at most $t$ and dominates $L$. Also, $w(S_L) \leq w(S_{L_x}) + w(S_{L_y}) - w_i \leq w'(L_x) + w'(L_y) - w_i = w'(L)$.
Therefore, the algorithm invariants hold for $L$.
\end{proof}

For a fixed pair $(p_x, p_{y})$, among all sublists $L(x,y,t')$, $2\leq t'\leq t-1$, we keep the one with minimum $w'(L(x,y,t'))$ and add it to $\calL_t(i)$. In this way, the bidirectional processing procedure for $p_i$ adds $O(n^2)$ sublists to $\calL_t(i)$. 

\subsubsection{Summary}
In summary, the $t$-th iteration of the algorithm computes $O(n^2)$ sublists for each $\calL_t(i)$, and in total computes $O(n^3)$ sublists for $\calL_t$.

After the $k$-th iteration, among all sublists of $\calL_k$ that are $P$, we find the one $L^*$ of minimum $w'(L^*)$ and return $S_{L^*}$ as our optimal dominating set of size at most $k$. If no sublist of $\calL_k$ is $P$, then we report that a dominating set of size at most $k$ does not exist. 

\subsection{Algorithm correctness}
\label{sec:correct}

In this section, we will prove the following lemma, which establishes the correctness of the algorithm. 
\begin{lemma}\label{lem:correctweight}
$S_{L^*}$ is an optimal dominating set and $W^*=w'(L^*)$.
\end{lemma}

According to our algorithm invariants, the sublist $L^*$, which is $P$, is dominated by $S_{L^*}$, with $|S_{L^*}| \leq k$ and $w(S_{L^*}) \leq w'(L^*)$. Therefore, we obtain $W^* \leq w(S_{L^*}) \leq w'(L^*)$. To prove the lemma, in the following we will prove that $w'(L^*) \leq W^*$. Note that our proof also argues that such a sublist $L^*\in \calL_k$ with $L^*=P$ must exist. 


Let $S$ be an optimal dominating set of size at most $k$, and let $\phi: \calA \rightarrow S$ be the line-separable assignment guaranteed by Lemma~\ref{lem:linesep}. Order sublists of $\calA=\langle\alpha_1,\alpha_2,\ldots,\alpha_m\rangle$ counterclockwise along $\calH(P)$ as discussed in Section~\ref{sec:overview} (we also follow the notation there, e.g., $\calA[g,h]$, $\calA_{p_i}$, $S[g,h]$). We will prove the following lemma. 

\begin{lemma}\label{lem:correct10}
For any point $p_i \in S$ and any $t$ with $1\leq t\leq k$, for any rank-$t$ sublist $\calA[g,h]$ of $p_i$, $\calL_t(i)$ must contain a sublist $L$ with $\calA[g,h]\subseteq L$ and $w'(L)\leq w(S[g,h])$.    
\end{lemma}

Lemma~\ref{lem:correct10} will lead to $w'(L^*) \leq W^*$ and thus prove Lemma~\ref{lem:correctweight}. Indeed, as discussed before, $p_i$ must have a rank-$k$ sublist $\calA[g,h]$ that is $P$ and $S[g,h]=S$. Since $\calL_k(i)$ contains a sublist $L$ with $\calA[g,h]\subseteq L$ and $w'(L)\leq w(S[g,h])$, we obtain that $L=P$ and $w'(L)\leq w(S)=W^*$. By definition, $L^*=P$ and $w'(L^*)\leq w'(L)\leq w(S)=W^*$. This proves Lemma~\ref{lem:correctweight}. 

In the rest of this section, we prove Lemma~\ref{lem:correct10}, by induction on $t=1,2,\ldots,k$. 
Recall that by Lemma~\ref{lem:onesizegroup} the base case starts from $t=1$ since $S$ has at least one point that has a rank-1 sublist in $\calA$. 

For the base case $t=1$, suppose that $p_i\in S$ has a rank-1 sublist $\calA[g,h]$. By definition, $\calA[g,h]$ is the only sublist in the group $\calA_{p_i}$, and it is the main sublist of $p_i$, meaning that $p_i\in \calA[g,h]$. Therefore, $w(S[g,h])=w_i$. Since $p_i\in \calA[g,h]$ and $P(b_i^i, a_i^i)$ is the maximal sublist of $P$ containing $p_i$ by the definitions of $a_i^i$ and $b_i^i$, we have $\calA[g,h]\subseteq P(b_i^i, a_i^i)$. According to our algorithm, $\calL_1(i)$ consists of a single sublist $L=P(b_i^i, a_i^i)$ with $S_L=\{p_i\}$ and $w'(S_L)=w_i$. Hence, Lemma~\ref{lem:correct10} holds for the base case $t=1$. 

We now assume that Lemma~\ref{lem:correct10} holds for every $t'\in [1,t-1]$, i.e., for every sublist $\calA[g',h']$ that is the rank-$t'$ sublist of some $p_i\in S$, $\calL_{t'}(i)$ contains a sublist $L'$ satisfying: $\calA[g',h'] \subseteq L'$ and $w'(L') \leq w(S[g',h'])$.
In the following, we prove that Lemma~\ref{lem:correct10} holds for $t$, i.e., for every sublist $\calA[g,h]$ that is the rank-$t$ sublist of some $p_i\in S$, $\calL_{t}(i)$ contains a sublist $L$ satisfying: $\calA[g,h] \subseteq L$ and $w'(L) \leq W(S[g,h])$. 

Suppose that $p_i$ is a rank-$t$ center of $\calA[g,h]$. 
By definition, at least one of the two end sublists of $\calA[g,h]$ is in the group $\calA_{p_i}$ and $\calA[g,h]$ contains the main sublist of $p_i$. Depending on whether the main sublist of $p_i$ is an end sublist of $\calA[g,h]$, there are two cases. 

\subsubsection{Case (1): The main sublist of $\boldsymbol{p_i}$ is an end sublist of $\boldsymbol{\calA[g,h]}$}

Without loss of generality, we assume that the clockwise end sublist of $\calA[g,h]$ is the main sublist of $p_i$, which is $\alpha_g$. We have the following lemma to ``decompose'' $\calA[g,h]$ into smaller sublists. See Figure~\ref{fig:sublist30}.

\begin{figure}[h]
\begin{minipage}[h]{\textwidth}
\begin{center}
\includegraphics[height=2.0in]{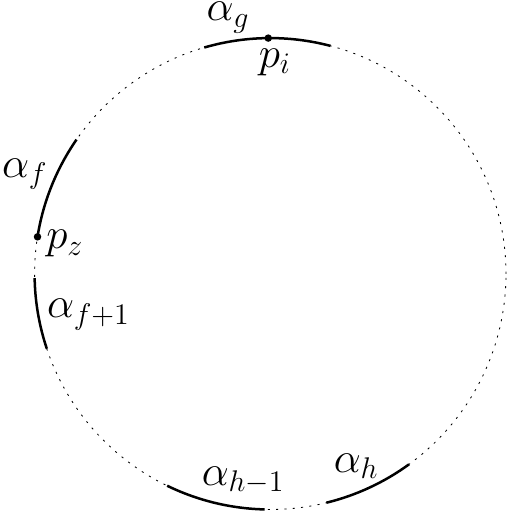}
\caption{\footnotesize Illustrating the sublists of $\calA$, represented by the black solid arcs (the dotted circle represents $\calH(P)$).}
\label{fig:sublist30}
\end{center}
\end{minipage}
\vspace{-0.1in}
\end{figure}

\begin{lemma}\label{lem:decomp}
\begin{description}
    \item[1. (The open case)] If $p_i$ is a rank-$t$ open center of $\calA[g,h]$, then there exist $t'\in [1,t-1]$, $f\in [g,h-1]$, and a point $p_{i'}\in S[g,h]$ such that $p_i$ is a rank-$t'$ center of $\calA[g,f]$ and $p_{i'}$ is a rank-$(t-t')$ center of $\calA[f+1,h]$. Furthermore, $w(S[g,h])=w(S[g,f])+w(S[f+1,h])$.
    \item[2. (The closed case)] If $p_i$ is a rank-$t$ closed center of $\calA[g,h]$, then there exist $t'\in [1,t-1]$, $f\in [g,h-1]$, and a point $p_{i'}\in S[g,h]$ such that $p_i$ is a rank-$t'$ center of $\calA[g,f]$ and $p_{i'}$ is a rank-$(t-t')$ center of $\calA[f+1,h-1]$. Furthermore, $w(S[g,h])=w(S[g,f])+w(S[f+1,h-1])$.
\end{description}
\end{lemma}
\begin{proof}
We first discuss the open case. In this case, the sublist $\alpha_{h}$ must be in $\calA_{p_i'}$ for another point $p_{i'}\in S[g,h]\setminus\{p_i\}$. Let $\alpha_{f}$ be the farthest counterclockwise sublist of $\calA[g,h]$ assigned to $p_i$. Let $t'=|S[g,f]|$. 
Since $\alpha_g$ is a sublist of $p_i$ (i.e., $\alpha_g\in \calA_{p_i}$), by the line-separable property, any point of $S[f+1,h]$ cannot have a sublist in $\calA[g,f]$. This implies $1\leq t'\leq t-1$ and $|S[f+1,h]|=t-t'$, 

Consider a point $p\in S[f+1,h]$. By definition, all sublists of $\calA_p$ are in $\calA[g,h]$. Since $\calA_p$ cannot have a sublist in $\calA[g,f]$, all sublists of $\calA_p$ are in $\calA[f+1,h]$. Since $\alpha_h\in \calA_{p_i'}$, $p_{i'}$ is a rank-($t-t'$) center of $\calA[f+1,h]$ by definition. 

On the other hand, since both end sublists of $\calA[g,f]$ are from $\calA_{p_i}$, again due to the line-separable property, none of the points of $S[g,f]\setminus\{p_i\}$ can have a sublist outside $\calA[g,f]$. As $|S[g,f]|=t'$, by definition, $p_{i}$ is a rank-($t'$) (closed) center of $\calA[g,f]$. 

The above discussion also implies that $S[g,f]$ and $S[f+1,h]$ form a partition of $S[g,h]$. Therefore, $w(S[g,h])=w(S[g,f])+w(S[f+1,h-1])$.

We now discuss the closed case. In this case, $\alpha_h$ is also assigned to $p_i$, i.e., $\alpha_h\in \calA_{p_i}$. Let $p_{i'}$ be the point such that $\alpha_{h-1}\in \calA_{p_i'}$.
We define $\alpha_{f}$ to be the farthest counterclockwise sublist in $\calA[g,h-1]$ assigned to $p_i$. 
Let $t'=|S[g,f]|$. As both $\alpha_g$ and $\alpha_f$ are in $\calA_{p_i}$, by the line-separable property, none of the points of $S[g,f]\setminus\{p_i\}$ can have a sublist outside $\calA[g,f]$. Therefore, $p_i$ is a rank-$t'$ (closed) center of $\calA[g,f]$. Similarly, since both $\alpha_f$ and $\alpha_h$ are in $\calA_{p_i}$, none of the points of $S[f+1,h-1]$ can have a sublist outside $\calA[f+1,h-1]$. By the definition of $f$, no sublist of $S[f+1,h-1]$ is in $\calA_{p_i}$. Therefore, $|S[f+1,h-1]|=t-t'$ and $p_{i'}$ is a rank-$(t-t')$ center of $\calA[f+1,h-1]$. The above discussion also implies that $S[g,f]$ and $S[f+1,h-1]$ form a partition of $S[g,h]$. Therefore, $w(S[g,h])=w(s[g,f])+w(s[f+1,h-1])$.
\end{proof}

Define $h'$ to be $h$ if $p_i$ is an open center, and to be $h-1$ if $p_i$ is a closed center. In either case, $p_{i'}$ is a rank-$(t-t')$ center of $\calA[f+1,h']$. Let $p_j$ be the counterclockwise endpoint of $\calA[f+1,h']$. 
Let $p_z$ be the counterclockwise endpoint of $\alpha_f$ in Lemma~\ref{lem:decomp} (see Figure~\ref{fig:sublist30}). 


Now, consider the $t$-th iteration of our algorithm when $p_i$ is processed during the \textit{counterclockwise processing procedure}, and the triple $(p_j, p_z, t')$ is examined. During this step, the algorithm constructs a sublist $L$ for $\calL_t(i)$ as the union of four sublists (see Figure~\ref{fig:sublist10}): $P(b_i^i,a_i^i)$, the minimum-value sublist $L_1$ of $\calL_{t'}(i)$ containing $P[i,z]$, the minimum-value sublist
$L_2$ of $\calL_{t - t'}$ containing $P[z_1+1, j]$, and $P(z_2,a_i^{z_2+1})$, where $p_{z_1}$ and $p_{z_2}$ are the counterclockwise endpoints of $L_1$ and $L_2$, respectively. Also, $S_L = S_{L_1} \cup S_{L_2}$ and $w'(L) = w'(L_1) + w'(L_2)$. The following lemma completes our proof for Case (1). 

\begin{lemma}\label{lem:correct20}
$\calA[g,h]\subseteq L$ and $w'(L)\leq w(S[g,h])$. 
\end{lemma}
\begin{proof}   
By the induction hypothesis, $\calL_{t'}(i)$ contains a sublist $L_1'$ such that $\calA[g,f]\subseteq L_1'$ and $w'(L_1')\leq w(S[g,f])$; similarly, $\calL_{t-t'}(i')$ contains a sublist $L_2'$ such that $\calA[f+1,h']\subseteq L_2'$ and $w'(L_2)\leq w(S[f+1,h'])$. See Figure~\ref{fig:sublist40}. 

\begin{figure}[t]
\begin{minipage}[h]{\textwidth}
\begin{center}
\includegraphics[height=2.9in]{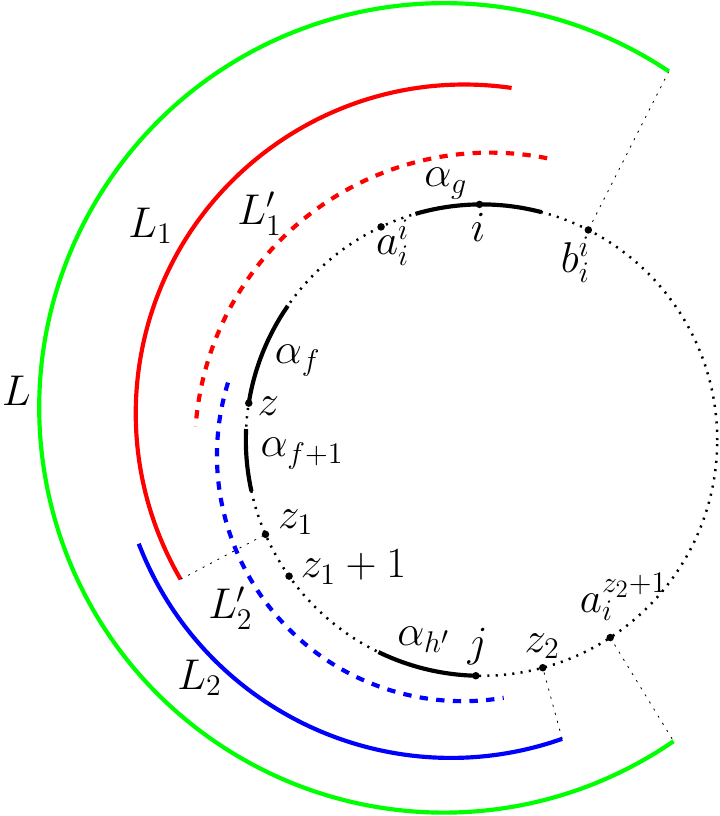}
\caption{\footnotesize Illustrating the proof of Lemma~\ref{lem:correct20}. For points of $P$, only their indices are shown. The red (resp., blue) dashed arc represents $L_1'$ (resp., $L_2'$). The red (resp., blue, green) solid arc represents $L_1$ (resp., $L_2$, $L$). The black solid arcs on the dotted circle represents sublists of $\calA$.}
\label{fig:sublist40}
\end{center}
\end{minipage}
\vspace{-0.1in}
\end{figure}

Since $P[i,z]\subseteq L_1$ and $p_i\in \alpha_g\subseteq P(b_i^i,a_i^i)$, we have $P[i,z]\subseteq\calA[g,f]\subseteq P(b_i^i,a_i^i)\cup L_1$. Since $\calA[g,f]\subseteq L_1'$, $P[i,z]\subseteq L_1'$. As $L_1$ is the minimum-value sublist of $\calL_{t'}(i)$ containing $P[i,z]$ and $L_1'$ is a sublist of $\calL_{t'}(i)$ containing $P[i,z]$, we obtain $w'(L_1)\leq w'(L_1')$. Since $w'(L_1')\leq w(S[g,f])$, it follows that
\begin{equation}\label{equ:10}
w'(L_1)\leq w(S[g,f]).    
\end{equation}


We claim that $p_{z_1+1}$ must be in a sublist of $\calA[f+1,h']$. Assume to the contrary that this is not true. Then, since $\calA[g,f]$ and $\calA[f+1,h']$ are consecutive and $\calA[g,f]\subseteq P(b_i^i,a_i^i)\cup L_1$, we have $\calA[g,h']=\calA[g,f]\cup \calA[f+1,h'] \subseteq P(b_i^i,a_i^i)\cup L_1$. 
By definition, $|S_{L_1}|\leq t'=|S[g,f]|<t$. 
As $p_i$ is in $S_{L_1}$ and $p_i$ dominates $P(b_i^i,a_i^i)$, if we replace points of $S[g,h']$ in $S$ by points of $S_{L_1}$, we can still obtain a new dominating set $S'$ for $P$ of size at most $k$. But since $w(S_{L_1})\leq w'(L_1)\leq w(S[g,f])$, the new dominating set $S'$ has weight $w(S')=w(S\setminus S[g,h']) +w(S_{L_1})<w(S\setminus S[g,h']) + w(S[g,f])+w(S[f+1,h']) = w(S)$, a contradiction to the optimality of $S$. Therefore, the above claim holds.

Due to the above claim, we have $P[z_1+1,j]\subseteq \calA[f+1,h']$. Since $\calA[f+1,h']\subseteq L_2'$, $P[z_1+1,j]\subseteq L_2'$.
Because $L_2$ is the minimum-value sublist of $\calL_{t - t'}$ containing $P[z_1+1, j]$ and $L_2'\in \calL_{t - t'}(i')\subseteq \calL_{t-t'}$, $w'(L_2)\leq w'(L_2')$ holds. As $w'(L_2')\leq w(S[f+1,h'])$, it follows that 
\begin{equation}\label{equ:20}
w'(L_2)\leq w(S[f+1,h']).    
\end{equation}

Recall that $L$ is the union of the following four sublists: $P(b_i^i,a_i^i)$, $L_1$, $L_2$, and $P(z_2,a_i^{z_2+1})$. Depending on whether $p_i$ is an open or closed rank-$t$ center of $\calA[g,h]$, there are two cases. 

\begin{itemize}
\item 
If $p_i$ is an open center, then $h'=h$ and thus $\calA[g,h]=\calA[g,f]\cup \calA[f+1,h']$. We showed above that $\calA[g,f]\subseteq P(b_i^i,a_i^i)\cup L_1$. As $\calA[f+1,h']$'s clockwise endpoint is $p_{z+1}$ and $p_z\in L_1$, we obtain $\calA[f+1,h']\subseteq L_1\cup P[z_1+1,j]$. Consequently, $\calA[g,h]\subseteq P(b_i^i,a_i^i)\cup L_1\cup P[z_1+1,j]$. As $P[z_1+1,j]\subseteq L_2$, it follows that $\calA[g,h]\subseteq P(b_i^i,a_i^i)\cup L_1\cup L_2\subseteq L$. 

By Lemma~\ref{lem:decomp}, $w(S[g,h])=w(S[g,f])+w(S[f+1,h'])$. As $w'(L)= w'(L_1) + w'(L_2)$, by Inequalities~\eqref{equ:10} and \eqref{equ:20}, we can derive $w'(L)\leq w(S[g,h])$.

\item 
If $p_i$ is a closed center, then $h'=h-1$ and $\calA[g,h]=\calA[g,f]\cup \calA[f+1,h']\cup \alpha_{h}$. 
As in the above open case, $\calA[g,f]\cup \calA[f+1,h']\subseteq P(b_i^i,a_i^i)\cup L_1\cup L_2\subseteq L$.

If $\calA[g,h]\subseteq P(b_i^i, a_i^i)\cup L_1\cup L_2$, then we still have $\calA[g,h]\subseteq L$. Otherwise, it must be the case that $p_{z_2+1}\in \alpha_{h}$. Consequently, since $p_i$ dominates $\alpha_{h}$ and $p_{j}$ is the counterclockwise endpoint of $\calA[f+1,h']$, by the definition of $a_i^{z_2+1}$, $\alpha_{h}\subseteq P[j+1,a_i^{z_2+1})$. As $\calA[g,f]\cup \calA[f+1,h']\subseteq P(b_i^i,a_i^i)\cup L_1\cup L_2$, we obtain $\calA[g,h]=\calA[g,f]\cup \calA[f+1,h']\cup \alpha_{h}\subseteq P(b_i^i,a_i^i)\cup L_1\cup L_2\cup\alpha_{h}\subseteq P(b_i^i,a_i^i)\cup L_1\cup L_2\cup P[j+1,a_i^{z_2+1})$. Since $p_j\in L_2$ by definition and $z_2$ is the counterclockwise endpoint of $L_2$, we have $L_2\cup P[j+1,a_i^{z_2+1})
\subseteq L_2\cup P(z_2,a_i^{z_2+1})$. Hence, 
$P(b_i^i,a_i^i)\cup L_1\cup L_2\cup P[j+1,a_i^{z_2+1})\subseteq P(b_i^i,a_i^i)\cup L_1\cup L_2\cup P(z_2,a_i^{z_2+1})=L$. Therefore, we obtain $\calA[g,h]\subseteq L$.

By Lemma~\ref{lem:decomp}, $w(S[g,h])=w(S[g,f])+w(S[f+1,h'])$.
As in the above case, since $w'(L)= w'(L_1) + w'(L_2)$, by Inequalities~\eqref{equ:10} and \eqref{equ:20}, we can finally derive $w'(L)\leq w(S[g,h])$. 
\end{itemize}

In summary, we have $\calA[g,h]\subseteq L$ and $w'(L)\leq w(S[g,h])$ in both cases. The lemma thus follows. 
\end{proof}

\subsubsection{Case (2): The main sublist of $\boldsymbol{p_i}$ is not an end sublist of $\boldsymbol{\calA[g,h]}$}

As $p_i$ is a rank-$t$ center of $\calA[g,h]$, one of the two end sublists of $\calA[g,h]$ must be in the group $\calA_{p_i}$. Without loss of generality, we assume that the counterclockwise end sublist of $\calA[g,h]$, i.e, $\alpha_h$, is in $\calA_{p_i}$. Note that $\alpha_h$ must be a secondary sublist of $p_i$. Let $\alpha_f$ be the main sublist of $p_i$. By definition, $\alpha_f\in \calA[g,h]$.
We have the following lemma. See Figure~\ref{fig:sublist50}.

\begin{lemma}\label{lem:decompbi}
There exists $t'\in [2,t-1]$ such that $p_i$ is a rank-$t'$ center of $\calA[f,h]$ and is a rank-$(t-t'+1)$ center of $\calA[g,f]$. Furthermore,  $w(S[g,h])=w(S[f,h]) + w(S[g,f]) - w_i$.
\end{lemma}
\begin{proof} 
Define $t'=|S[f,h]|$. As both $\alpha_{h}$ and $\alpha_{f}$ are sublists of $\calA_{p_i}$, $S[f,h]\setminus\{p_i\}$ is not empty and thus $t'\geq 2$, and furthermore, 
due to the line-separable property, no center of $S[f,h]\setminus\{p_i\}$ can have a sublist outside $\calA[f,h]$. Consequently, $p_i$ is a rank-$t'$ center of $\calA[f,h]$.

Consider any center $p_{i'}\in S[g,f]\setminus\{p_i\}$. We claim that all sublists of $\calA_{p_{i'}}$ are in $\calA[g,f]$. Indeed, since $p_i$ is a rank-$t$ center in $\calA[g,h]$ and $p_{i'}\in S[g,h]$, all sublists of $\calA_{p_{i'}}$ are in $\calA[g,h]$. As both $\alpha_{h}$ and $\alpha_{f}$ are sublists of $p_i$, due to the line-separable property, $\calA_{p_{i'}}$ cannot have a sublist in $\calA[f,h]$. Therefore, all sublists of $\calA_{p_{i'}}$ must be in $\calA[g,f]$. 

In light of the above claim, since $p_i$ is both in $S[g,f]$ and $S[f,h]$, we have $|S[g,f]|=t+1-|S[f,h]|=t+1-t'$. Since $\alpha_f$ is an end sublist of $\calA[g,f]$, $p_i$ is a rank-$(t+1-t')$ center of $\calA[g,f]$. 

The above discussion also implies that $S[f,h]\setminus\{p_i\}$ and $S[g,f]\setminus\{p_i\}$ form a partition of $S[g,h]\setminus\{p_i\}$. As $p_i$ is in $S[f,h]$, $S[g,f]$, and $S[g,h]$, we have $w(S[g,h])=w(S[f,h]) + w(S[g,f]) - w_i$.
\end{proof}

\begin{figure}[t]
\begin{minipage}[h]{\textwidth}
\begin{center}
\includegraphics[height=2.0in]{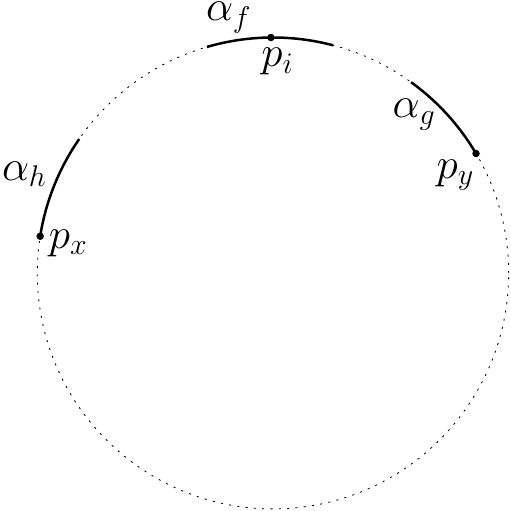}
\caption{\footnotesize Illustrating the sublists of $\calA$, represented by the black solid arcs (the dotted circle represents $\calH(P)$).}
\label{fig:sublist50}
\end{center}
\end{minipage}
\vspace{-0.1in}
\end{figure}



Let $t'$ be defined by Lemma~\ref{lem:decompbi}. 
Let $p_x$ and $p_y$ denote the counterclockwise and clockwise endpoints of $\calA[g,h]$, respectively; see Figure~\ref{fig:sublist50}.

Now consider the $t$-th iteration of the algorithm in which $p_i$ is processed via the 
\textit{bidirectional processing procedure}, and the triple 
$(p_x, p_{y}, t')$ is considered. 
In this step, the algorithm constructs a sublist $L \in \calL_t(i)$ as the union of the three sublists: $P(b_i^i, a_i^i)$,  the minimum-value enclosing sublist $L_x$ in $\calL_{t'}(i)$ that contains $P[i,x]$, and the minimum-value enclosing sublist 
$L_y$ in $\calL_{t+1-t'}(i)$ that contains $P[y,i]$. Also, $S_L = S_{L_x} \cup S_{L_y}$ and 
$w'(L) = w'(L_x) + w'(L_y) - w_i$. The following lemma completes our proof for Case (2). 

\begin{lemma}\label{lem:correct30}
$\calA[g,h]\subseteq L$ and $w'(L)\leq w(S[g,h])$. 
\end{lemma}
\begin{proof}    
By the induction hypothesis on $t'$, $\calL_{t'}(i)$ contains a sublist 
$L_x'$ such that $\calA[f,h] \subseteq L_x'$ and $w'(L_x') \leq w(S[f,h])$. 
Similarly, $\calL_{t+1-t'}(i)$ contains a sublist $L_y'$ such that 
$\calA[g,f] \subseteq L_y'$ and $w'(L_y') \leq w(S[g,f])$. See Figure~\ref{fig:sublist60}.

\begin{figure}[t]
\begin{minipage}[h]{\textwidth}
\begin{center}
\includegraphics[height=2.3in]{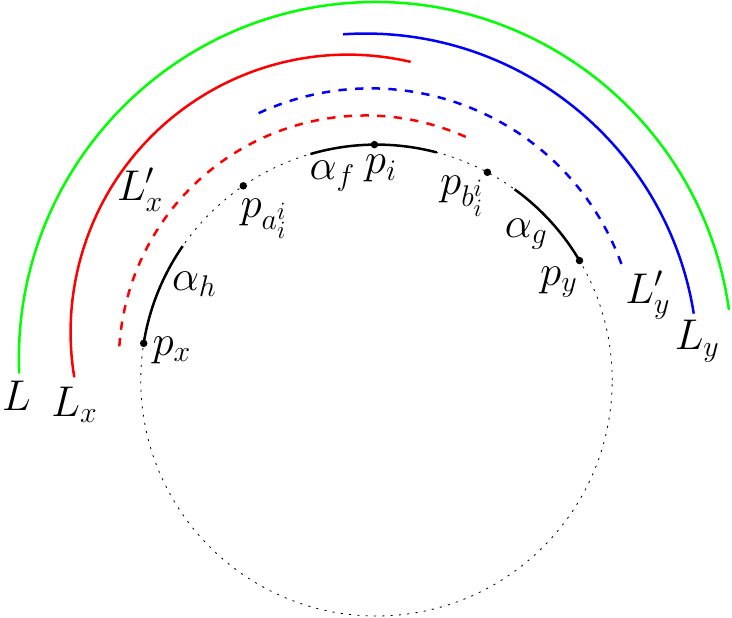}
\caption{\footnotesize Illustrating the proof of Lemma~\ref{lem:correct30}. The red (resp., blue) dashed arc represents $L_1'$ (resp., $L_2'$). The red (resp., blue, green) solid arc represents $L_1$ (resp., $L_2$, $L$). The black solid arcs on the dotted circle represents sublists of $\calA$.}
\label{fig:sublist60}
\end{center}
\end{minipage}
\vspace{-0.1in}
\end{figure}

Recall that $\alpha_f$ is the main sublist of $p_i$. Since $p_i\in \alpha_f$, by the definitions of $b_i^i$ and $a_i^i$, $\alpha_f\subseteq P(b_i^i,a_i^i)$ holds. Since $P[i,x]\subseteq L_x$, we have $\calA[f,h] \subseteq P(b_i^i,a_i^i) \cup L_x$.  
Moreover, $P[i,x]\subseteq \calA[f,h]$ by definition. As $\calA[f,h]\subseteq L_x'$, we have $P[i,x]  \subseteq L_x'$.  
Since $L_x'\in \calL_{t'}(i)$ and $L_x$ is the minimum-value sublist in $\calL_{t'}(i)$ that contains $P[i,x]$, it follows that $w'(L_x) \leq w'(L_x')$. Because $w'(L_x') \leq w(S[f,h])$, we obtain $w'(L_x) \leq w(S[f,h])$.


By a symmetric argument, we can obtain  $\calA[g,f] \subseteq P(b_i^i,a_i^i) \cup L_y$ and 
$w'(L_y) \leq w(S[g,f])$.


Combining the above leads to $\calA[g,h] = \calA[f,h] \cup \calA[g,f] \subseteq P(b_i^i,a_i^i) \cup L_x \cup L_y = L$ and  
$w'(L) = w'(L_x) + w'(L_y) - w_i \leq w(S[f,h]) + w(S[g,f]) - w_i = w(S[g,h])$, where the last equality is due to Lemma~\ref{lem:decompbi}. The lemma thus follows. 
\end{proof}

\subsection{Time analysis and implementation}
\label{sec:time}

In the $t$-th iteration, $1\leq t\leq k$, the algorithm performs $O(tn^3)$ minimum-value enclosing sublist queries: Given a sublist $L$ of $P$, compute the minimum value sublist containing $L$ in a set $\calL$ of sublists.  
This problem has been studied in \cite{ref:TkachenkoDo25}. We summarize the result in the following lemma. 

\begin{lemma}{\em(\cite{ref:TkachenkoDo25})}
\label{lem:enclosequery}
Given a set $\calL$ of $m$ sublists of $P$, each associated with a value, we can construct a data structure for $\calL$ in $O(m\log m)$ time, with $m=|\calL|$, so that each minimum-value enclosing sublist query can be answered in $O(\log^2 m)$ time.     
\end{lemma}

In our algorithm, the total size of the involved sets $\calL$ for which we need to construct the above data structure is $O(n^3)$ in each iteration. Hence, the total time spent on these queries (including the time for constructing the data structures) in the $t$-th iteration of the algorithm is $O(tn^3\log^2 n)$ time. 

In addition, in the $t$-th iteration, the algorithm performs $O(tn^2)$ operations to compute the indices $a_i^j$ and $b_i^j$. 
To this end, the following lemma provides a data structure. 

\begin{lemma}
\label{lem:firstout}
We can construct a data structure for $P$ in $O(n\log^2 n)$ time such that the indices $a_i^j$ and $b_i^j$ can be computed in $O(\log^2 n)$ time for any two points $p_i, p_j\in P$.
\end{lemma}
\begin{proof}
The algorithm is similar to an algorithm for the unit-disk case given in \cite{ref:TkachenkoDo25} except that we replace the farthest point Voronoi diagram by the farthest disk Voronoi diagram~\cite{ref:RappaportCo89}. For completeness, we present the details in the following. 

We start with constructing a complete binary tree $T$ whose leaves from left to right store $p_1,p_2,\ldots,$ $p_n$, respectively. For each node $v\in T$, let $P_v$ denote the subset of points of $P$ in the leaves of the subtree rooted at $v$. We construct the farthest disk Voronoi diagram $\fdvd(v)$ on $\calD(P_v)$~\cite{ref:RappaportCo89}, i.e., the set of disks centered at the points of $P_v$. We also build a point location data structure on $\fdvd(v)$~\cite{ref:EdelsbrunnerOp86,ref:KirkpatrickOp83}. Constructing $\fdvd(v)$ can be done in $O(|P_v|\log|P_v|)$ time~\cite{ref:RappaportCo89}. Constructing the point location data structure takes $O(|P_v|)$ time~\cite{ref:EdelsbrunnerOp86,ref:KirkpatrickOp83}. Doing this for all nodes $v\in T$ takes $O(n\log^2 n)$ time.

Given two points $p_i,p_j\in P$, we can compute the index $a_i^j$ in $O(\log^2 n)$ time as follows. 
We first check whether $a_i^j=0$, i.e., whether $|D_{p_i}D_p|\leq 0$ for all points $p\in P$. For this, we compute the farthest disk $D_{p'}$ of $\calD(P)$ from $p_i$, by finding the cell of $\fdvd(v)$ containing $p_i$ for the root $v$ of $T$. Notice that $a_i^j=0$ if and only if $|D_{p'}D_{p_i}|\leq 0$. Next, we check if $a_i^j=j$ by testing whether $|D_{p_i}D_{p_j}|>0$ holds. In the following, we assume that $a_i^j$ is not equal to $0$ or $j$. There are two cases depending on whether $a_i^j$ is larger or smaller than $j$. 

\begin{itemize}
\item 
We first consider the case $a_i^j>j$. Let $v_j$ be the leaf of $T$ storing $p_j$. Starting from $v_j$, we go up on $T$ until the first node $v$ whose right child $u$ has the following property: The distance between $p_i$ and its farthest disk $D_{p'}$ for $p'\in P_u$ is greater than $r_i$ (meaning that $|D_{p_i}D_{p'}|>0$). Let $v^*$ denote such a node $v$. Specifically, starting from $v=v_j$, if $v$ does not have a right child, then set $v$ to its parent. Otherwise, let $u$ be the right child of $v$. Using a point location query on $\fdvd(P_u)$, we find the farthest disk $D_{p'}$ of $p_i$ with $p'\in P_u$. If $|D_{p'}D_{p_i}|\leq 0$, then we set $v$ to its parent; otherwise, we set $v^*=v$. Since $a_i^j>j$, $v^*$ is guaranteed to be found. Next, starting from the right child of $v^*$, we perform a top-down search process. For each node $v$, let $u$ be its left child. We find the farthest disk $D_{p'}$ of $p_i$ with $p'\in P_u$. If $|D_{p'}D_{p_i}|> 0$, then we set $v$ to $u$; otherwise, we set $v$ to its right child. The process will eventually reach a leaf $v$, which stores the point $a_i^j$. The total time is $O(\log^2 n)$ as the search process calls point location queries $O(\log n)$ times and each point location query takes $O(\log n)$ time. 

\item 
For the case $a_i^j<j$, $D_{p_i}$ intersects the disk centered at the point in any leaf to the right of $v_j$ and thus the bottom-up procedure in the above algorithm will eventually reach the root. Note that $a_i^j$ is the first point in $P[1,j-1]$ whose disk does not intersect $D_{p_i}$. We proceed with the following. If $|D_{p_1}D_{p_i}|> 0$, then $a_i^j=1$ and we can stop the algorithm. Otherwise, starting from $v=v_1$, the leftmost leaf of $T$, we apply the same algorithm as in the above case, i.e., first run a bottom-up procedure and then a top-down one. The total time of the algorithm is $O(\log^2 n)$. 
\end{itemize}

Computing $b_i^j$ can be done similarly in $O(\log^2 n)$ time. 
\end{proof}

In summary, the $t$-th iteration of the algorithm takes $O(tn^3\log^2 n)$ time. As there are $k$ iterations, the total time of the algorithm is $O(k^2n^3\log^2 n)$. We thus have the following theorem.

\begin{theorem}\label{thm:domsetwgt}
Given a number $k$ and a set of $n$ weighted disks whose centers are in convex position in the plane, we can find in $O(k^2n^3\log^2 n)$ time a minimum-weight dominating set of size at most $k$ in the disk graph, or report that no such dominating set exists. 
\end{theorem}

Applying Theorem~\ref{thm:domsetwgt} with $k=n$ leads to the following result. 
\begin{corollary}\label{coro:domsetwgt}
Given a set of $n$ weighted disks whose centers are in convex position in the plane, we can compute a minimum-weight dominating set in the disk graph in $O(n^5\log^2 n)$ time. 
\end{corollary}

\section{The unweighted case}
\label{sec:domunwgt}

In this section, we consider the unweighted case. The goal is to compute a smallest dominating set in the disk graph $G(P)$. 
By setting the weights of all points of $P$ to $1$ and applying Corollary~\ref{coro:domsetwgt}, one can solve the unweighted problem in $O(n^5\log^2 n)$ time. In the following, we provide an algorithm of $O(k^2n\log n)$ time, where $k$ is the size of the smallest dominating set.

We follow the dynamic programming scheme of the weighted case, but incorporate a greedy strategy by taking advantage of the property that all points of $P$ have the same weight. Section~\ref{sec:unweightdescription} first describes the algorithm. The correctness proof is given in Section~\ref{sec:unweightcorrect}. Section~\ref{sec:unweighttime} finally discusses the implementation and the time analysis. 

\subsection{Algorithm description}
\label{sec:unweightdescription}

We still proceed in iterations $t=1,2\ldots$.
In each $t$-th iteration, we compute for each $p_i\in P$ a set $\calL_t(i)$ of $O(t)$ sublists, with each sublist $L$ associated with a subset $S_L\subseteq P$, such that the following algorithm invariants are maintained: (1) $S_L$ dominates $L$; (2) $p_i\in S_L$; (3) $p_i\in L$; (4) $|S_L|\leq t$. Define $\calL_t=\cup_{p_i\in P}\calL_t(i)$. Hence, $|\calL_t|=O(nt)$.

Let $k$ be the smallest dominating set size. We will show that after $k$ iterations, $\calL_k$ is guaranteed to contain a sublist that is $P$. As such, if a sublist that is $P$ is computed for the first time in the algorithm, then we can stop the algorithm.

Initially, $t=1$, and our algorithm does the following for each $p_i\in P$. We compute $L=P(b_i^i, a_i^i)$, and set $S_L=\{p_i\}$. We let $\calL_1(i)=\{L\}$. Obviously, all algorithm invariants hold for $L$. 

In general, suppose that for each $t'\in [1, t - 1]$, we have computed a collection $\calL_{t'}(i)$ of $O(t')$ sublists for each $p_i\in P$, with each sublist $L' \in \calL_{t'}(i)$ associated with a subset $S_{L'} \subseteq P$, such that the algorithm invariants hold, i.e., $L'$ is dominated by $S_{L'}$, $p_i\in S_{L'}$, $p_i\in L'$, and $|S_{L'}|\leq t'$.

The $t$-th iteration works as follows. For each point $p_i \in P$, we perform counterclockwise/clockwise processing procedures as well as bidirectional preprocessing procedures. Note that, unlike the weighted case, these procedures now incorporate greedy strategies. 
The details are given below. 

\subsubsection{Counterclockwise/clockwise processing procedures}
For each point $p_i\in P$, for each $t'\in [1,t-1]$, we perform the following {\em counterclockwise processing procedure}. 
One can still refer to Figure~\ref{fig:sublist10} (but the notation $z$ in the figure can be ignored). 

By our algorithm invariants for $t'$, every sublist of $\calL_{t'}(i)$ contains $p_i$. 
We first find the sublist of $\calL_{t'}(i)$ whose counterclockwise endpoint is farthest from $p_i$ (in the counterclockwise direction). Let $L_1$ denote the list and let $p_{z_1}$ be the counterclockwise endpoint of $L_1$.
Note that finding $L_1$ can be done in $O(1)$ time if the sublist of $\calL_{t'}(i)$ whose counterclockwise endpoint farthest from $p_i$ is explicitly maintained
(e.g., in the $t'$-th procedure, whenever we add a list to $\calL_{t'}(i)$, we update the sublist whose counterclockwise endpoint is farthest from $p_i$).
We then perform a {\em counterclockwise farthest enclosing sublist query} on $\calL_{t-t'}$ to 
compute a sublist containing $p_{z_1+1}$ such that its counterclockwise endpoint is farthest from $p_{z_1+1}$ (in the counterclockwise direction).
Let $L_2$ be the sublist and $z_2$ the counterclockwise endpoint of $L_2$.

Next, we compute the index $a_{i}^{z_2+1}$. 
By construction, the union of the following four sublists is a (contiguous) sublist of $P$: $P(b_i^i,a_i^i)$, $L_1$, $L_2$, and $P(z_2, a_{i}^{z_2+1})$. Denote this combined sublist by $L(t')$, or simply $L$.  We set $S_L= S_{L_1} \cup S_{L_2}$.

\begin{observation}
All algorithm invariants hold for $L$. 
\end{observation}
\begin{proof}   
We argue that the algorithm invariants hold for $L$, i.e., $L$ is dominated by $S_L$, $p_i\in S_L$, $p_i\in L$, and $|S_L|\leq t$. 
Indeed, by the algorithm invariants for $t'$, $L_1$ is dominated by $S_{L_1}$, $p_i \in S_{L_1}$, $p_i\in L_1$, and $|S_{L_1}| \leq t'$. Since $L_2 \in \calL_{t - t'}$, by the algorithm invariants for $t-t'$, $L_2$ is dominated by $S_{L_2}$ and $|S_{L_2}| \leq t - t'$. In addition, by definition, the two sublists $P(b_i^i, a_i^i)$ and $P(z_2, a_i^{z_2+1})$ are dominated by $p_i$. Combining the above, the algorithm invariants hold for $L$. 
\end{proof}

Among the $O(t)$ sublists $L(t')$ computed above for all $t'\in [1,t-1]$, we keep the one $L$ whose counterclockwise endpoint is farthest from $p_i$ and add it to $\calL_t(i)$. 

The counterclockwise processing procedure described above for $p_i$ adds at most one sublist to $\calL_t(i)$. Symmetrically, we also perform a clockwise processing procedure for $p_i$, which also contributes at most one sublist to $\calL_t(i)$. 



\subsubsection{Bidirectional processing procedures}
For each point $p_i \in P$,  for each $t'\in [2, t-1]$, we perform the following bidirectional processing procedure. 
One can still refer to Figure~\ref{fig:sublist20} (but the notation $x$ and $y$ in the figure can be ignored). 

We find the sublist $L_x$ in $\calL_{t'}(i)$ whose counterclockwise endpoint is farthest from $p_i$ (along the counterclockwise direction), and also find the sublist $L_y$ in $\calL_{t+1-t'}(i)$ whose clockwise endpoint is farthest from $p_i$ (along the clockwise direction). 
The following three sublists form a (contiguous) sublist of $P$: $P(b_i^i,a_i^i)$, $L_x$, and $L_y$. Let $L$ be the union of the above three sublists. We set $S_L = S_{L_x} \cup S_{L_y}$, and add $L$ to $\calL_t(i)$.


\begin{observation}
All algorithm invariants hold for $L$. 
\end{observation}
\begin{proof}    
We argue that all algorithm invariants hold for $L$. By the algorithm invariants for $t'$ and $t+1-t'$, $L_x$ is dominated by $S_{L_x}$, $L_y$ is dominated by $S_{L_y}$, $|S_{L_x}|\leq t'$, $|S_{L_x}|\leq t+1-t'$, $p_i$ is in both $S_{L_x}$ and $S_{L_y}$, and $p_i$ is also in both $L_x$ and $L_y$. Hence, $S_L$ dominates $L$, $|S_L|\leq t$, $p_i\in S_L$, and $p_i\in L$. Therefore, the algorithm invariants hold for $L$. 
\end{proof}

The above adds $O(t)$ sublists to $\calL_t(i)$.

\subsubsection{Summary}

The $t$-th iteration computes $O(t)$ sublists for $\calL_t(i)$ for each point $p_i\in P$. Hence, a total of $O(tn)$ sublists are added to $\calL_t$ in the $t$-th iteration. If any sublist $L\in \calL_t$ is $P$, then we stop the algorithm and return $S_L$ as a smallest dominating set. Otherwise, we continue on the next iteration. 

\subsection{Algorithm correctness}
\label{sec:unweightcorrect}

We prove the following lemma, which establishes the correctness of the algorithm. 
\begin{lemma}\label{lem:unweightcorrect}
If the algorithm first time computes a sublist $L$ that is $P$, then $S_L$ is a smallest dominating set of $G(P)$. 
\end{lemma}


Let $k$ be the smallest dominating set size for $G(P)$. It suffices to show that the algorithm will stop within $k$ iterations, i.e., there exists a sublist $L\in \calL_k$ such that $L=P$. 
Many arguments are similar to the weighted case proof in Section~\ref{sec:correct}. One difference is that we need to argue that our greedy strategy works. 

Let $S$ be a smallest dominating set (thus $|S|=k$), and let $\phi: \calA \rightarrow S$ be the line-separable assignment guaranteed by Lemma~\ref{lem:linesep}. Order sublists of $\calA=\langle\alpha_1,\alpha_2,\ldots,\alpha_m\rangle$ counterclockwise along $\calH(P)$ as before (we also follow the notation before, e.g., $\calA[g,h]$, $\calA_{p_i}$, $S[g,h]$). We will prove the following lemma. 

\begin{lemma}\label{lem:unweightcorrect10}
For any point $p_i \in S$ and any $t$ with $1\leq t\leq k$, for any rank-$t$ sublist $\calA[g,h]$ of $p_i$, $\calL_t(i)$ must contain a sublist $L$ with $\calA[g,h]\subseteq L$.    
\end{lemma}

Lemma~\ref{lem:unweightcorrect10} will lead to Lemma~\ref{lem:unweightcorrect}. Indeed, as discussed before, $P$ must be a rank-$k$ sublist $\calA[g,h]$ of $p_i$ with $\calA[g,h]=P$ and $S[g,h]=S$. Since $\calL_k(i)$ contains a sublist $L$ with $\calA[g,h]\subseteq L$ by Lemma~\ref{lem:unweightcorrect10}, we obtain that $L=P$. Therefore, our algorithm will find an optimal dominating set of size $k$ in the $k$-th iteration. 

In the rest of this section, we prove Lemma~\ref{lem:unweightcorrect10}, by induction on $t=1,2,\ldots,k$. 

For the base case $t=1$, suppose that $p_i\in S$ has a rank-1 sublist $\calA[g,h]$. By definition, $\calA[g,h]$ is the only sublist in the group $\calA_{p_i}$, and it is the main sublist of $p_i$. Since $p_i\in \calA[g,h]$ and $P(b_i^i, a_i^i)$ is the maximal sublist of $P$ containing $p_i$, we have $\calA[g,h]\subseteq P(b_i^i, a_i^i)$. According to our algorithm, $\calL_1(i)$ consists of a single sublist $L=P(b_i^i, a_i^i)$ with $S_L=\{p_i\}$. Hence, Lemma~\ref{lem:unweightcorrect10} holds for $t=1$. 

We now assume that Lemma~\ref{lem:unweightcorrect10} holds for every $t'\in [1,t-1]$, i.e., for every sublist $\calA[g',h']$ that is the rank-$t'$ sublist of some $p_i\in S$, $\calL_{t'}(i)$ contains a sublist $L'$ satisfying $\calA[g',h'] \subseteq L'$.
In the following, we prove that the lemma holds for $t$, i.e., for every sublist $\calA[g,h]$ that is the rank-$t$ sublist of some $p_i\in S$, $\calL_{t}(i)$ contains a sublist $L$ satisfying $\calA[g,h] \subseteq L$. 

Suppose that $p_i$ is a rank-$t$ center of a sublist $\calA[g,h]$. 
By definition, at least one of the two end sublists of $\calA[g,h]$ is in the group $\calA_{p_i}$ and $\calA[g,h]$ contains the main sublist of $p_i$. Depending on whether the main sublist of $p_i$ is an end sublist of $\calA[g,h]$, there are two cases. 

\subsection{Case (1): The main sublist of $\boldsymbol{p_i}$ is an end sublist of $\boldsymbol{\calA[g,h]}$}
Without loss of generality, we assume the clockwise end sublist of $\calA[g,h]$ is the main sublist of $p_i$, which is $\alpha_g$. 
Lemma~\ref{lem:decomp} still holds. 

Define $h'$ to be $h$ if $p_i$ is an open center and to be $h-1$ if $p_i$ is a closed center. Hence, in either case $p_{i'}$ is always a rank-$(t-t')$ center of $\calA[f+1,h']$, where $p_{i'}$ and $f$ are from Lemma~\ref{lem:decomp}.  
Let $p_j$ be the counterclockwise endpoint of $\calA[f+1,h']$. 

Now, consider the $t$-th iteration of our algorithm when $p_i$ is processed during the counterclockwise processing procedure and the value $t'$ is examined. During this step, the algorithm constructs a sublist $L$ for $\calL_t(i)$ as the union of four sublists: $P(b_i^i,a_i^i)$, a sublist $L_1$ of $\calL_{t'}(i)$ whose counterclockwise endpoint is farthest from $p_i$, a sublist $L_2$ of $\calL_{t-t'}$ containing $p_{z_1+1}$ whose counterclockwise endpoint is farthest from $p_{z_1+1}$, and $P[z_2,a_i^{z_2+1})$, where $p_{z_1}$ and $p_{z_2}$ are the counterclockwise endpoints of $L_1$ and $L_2$, respectively. Also, $S_L = S_{L_1} \cup S_{L_2}$. The following lemma completes our proof for Case (1). 

\begin{lemma}
$\calA[g,h]\subseteq L$. 
\end{lemma}
\begin{proof}   
By the induction, $\calL_{t'}(i)$ contains a sublist $L_1'$ such that $\calA[g,f]\subseteq L_1'$; similarly, $\calL_{t-t'}(i')$ contains a sublist $L_2'$ such that $\calA[f+1,h']\subseteq L_2'$. One can still refer to Figure~\ref{fig:sublist40} (but the notation $z$ in the figure can be ignored).

Since $L_1'\in \calL_{t'}(i)$ and $L_1$ is the sublist of $\calL_{t'}(i)$ whose counterclockwise endpoint is farthest from $p_i$, we have $L_1'\subseteq P(b_i^i,a_i^i)\cup L_1$. As $\calA[g,f]\subseteq L_1'$, we obtain $\calA[g,f]\subseteq P(b_i^i,a_i^i)\cup L_1$. 

We claim that $p_{z_1+1}$ must be in a sublist of $\calA[f+1,h']$. The argument follows the same as in Lemma~\ref{lem:correct20} (e.g., by assuming all point weights in that lemma are the same). Consequently, $P[z_1+1,j]\subseteq \calA[f+1,h']$. Since $\calA[f+1,h']\subseteq L_2'$, we have $P[z_1+1,j]\subseteq L_2'$.
Since $L_2$ is the sublist of $\calL_{t-t'}$ containing $p_{z_1+1}$ whose counterclockwise endpoint is farthest from $p_{z_1+1}$, and $L_2'$ is in $\calL_{t-t'}$ and also contains $p_{z_1+1}$, we obtain that $P[z_1+1,j]\subseteq L_2$. Since $\calA[g,f]\subseteq P(b_i^i,a_i^i)\cup L_1$, $z_1$ is the counterclockwise endpoint of $L_1$, and $j$ is the counterclockwise endpoint of $\calA[g,h']$, we can derive $\calA[g,h']=\calA[g,f]\cup \calA[f+1,h']\subseteq P(b_i^i,a_i^i)\cup L_1\cup P[z_1+1,j]\subseteq P(b_i^i,a_i^i)\cup L_1\cup L_2$. 

Recall that $L$ is the union of the following four sublists: $P(b_i^i,a_i^i)$, $L_1$, $L_2$, and $P(z_2,a_i^{z_2+1})$.
Depending on whether $p_i$ is an open or closed rank-$t$ center, there are two cases. 

\begin{itemize}
\item 
If $p_i$ is an open center, then $h'=h$ and thus $\calA[g,h]=\calA[g,h']$. We have showed above that $\calA[g,h']\subseteq P(b_i^i,a_i^i)\cup L_1\cup L_2\subseteq L$.  Therefore, $\calA[g,h]\subseteq L$ holds. 

\item 
If $p_i$ is a closed center, then $h'=h-1$ and $\calA[g,h]=\calA[g,h']\cup \alpha_{h}$. 

If $\calA[g,h]\subseteq P(b_i^i, a_i^i)\cup L_1\cup L_2$, then we still have $\calA[g,h]\subseteq L$. Otherwise, it must be the case that $p_{z_2+1}\in \alpha_{h}$. Consequently, since $p_i$ dominates $\alpha_{h}$ and $p_{j}$ is the counterclockwise endpoint of $\calA[g,h']$, by the definition of $a_i^{z_2+1}$, $\alpha_{h}\subseteq P[j+1,a_i^{z_2+1})$. As $\calA[g,h']\subseteq P(b_i^i,a_i^i)\cup L_1\cup L_2$, we obtain $\calA[g,h]=\calA[g,h']\cup \alpha_{h}\subseteq P(b_i^i,a_i^i)\cup L_1\cup L_2\cup\alpha_{h}\subseteq P(b_i^i,a_i^i)\cup L_1\cup L_2\cup P[j+1,a_i^{z_2+1})$. Since $p_j\in L_2$ and $z_2$ is the counterclockwise endpoint of $L_2$, we have $L_2\cup P[j+1,a_i^{z_2+1})
\subseteq L_2\cup P(z_2,a_i^{z_2+1})$. Hence, 
$P(b_i^i,a_i^i)\cup L_1\cup L_2\cup P[j+1,a_i^{z_2+1})\subseteq P(b_i^i,a_i^i)\cup L_1\cup L_2\cup P(z_2,a_i^{z_2+1})=L$. Therefore, we obtain $\calA[g,h]\subseteq L$.
\end{itemize}

In summary, $\calA[g,h]\subseteq L$ holds in both cases. The lemma thus follows. 
\end{proof}

\subsubsection{Case (2): The main sublist of $\boldsymbol{p_i}$ is not an end sublist of $\boldsymbol{\calA[g,h]}$}

As $p_i$ is a rank-$t$ center of $\calA[g,h]$, one of the two end sublists of $\calA[g,h]$ must be in the group $\calA_{p_i}$. With loss of generality, we assume that the counterclockwise end sublist of $\calA[g,h]$, i.e., $\alpha_h$, is in $\calA_{p_i}$. Note that $\alpha_h$ must be a secondary sublist of $p_i$. Lemma~\ref{lem:decompbi} still applies here. We let $t'$ and $f$ be from Lemma~\ref{lem:decompbi}. 

Now consider the $t$-th iteration of the algorithm in which $p_i$ is processed via the bidirectional processing procedure and the value $t'$ is examined. 
In this step, the algorithm constructs a sublist $L \in \calL_t(i)$ as the union of three sublists: $P(b_i^i, a_i^i)$,  the sublist $L_x$ of $\calL_{t'}(i)$ whose counterclockwise endpoint is farthest from $p_i$, and the sublist $L_y$ of $\calL_{t+1-t'}(i)$ whose clockwise endpoint is farthest from $P$. Also, $S_L = S_{L_x} \cup S_{L_y}$. The following lemma completes our proof for Case (2). 

\begin{lemma}
$\calA[g,h]\subseteq L$. 
\end{lemma}
\begin{proof}    
By the induction hypothesis on $t'$, $\calL_{t'}(i)$ contains a sublist $L_x'$ such that $\calA[f,h] \subseteq L_x'$; similarly, $\calL_{t+1-t'}(i)$ contains a sublist $L_y'$ such that $\calA[g,f] \subseteq L_y'$. 
One can still refer to Figure~\ref{fig:sublist60} (but the notation $p_x$ and $p_y$ in the figure can be ignored).

As $L_x$ is the sublist of $\calL_{t'}(i)$ whose counterclockwise endpoint is farthest from $p_i$ and $L_x'\in \calL_{t'}(i)$, we obtain $L_x'\subseteq P(b_i^i,a_i^i)\cup L_x$. Similarly, we have $L_y'\subseteq P(b_i^i,a_i^i)\cup L_y$.
Consequently, since $\calA[f,h] \subseteq L_x'$ and $\calA[g,f] \subseteq L_y'$, we can derive $\calA[g,h]=\calA[g,f]\cup \calA[f,h]\subseteq L_x'\cup L_y'\subseteq L_x\cup L_y\cup P(b_i^i,a_i^i)=L$. 
\end{proof}

\subsection{Time analysis and implementation}
\label{sec:unweighttime}

In each $t$-th iteration with $1\leq t\leq k$, we perform $O(tn)$ operations for computing indices $a_i^j$ and $b_i^j$ and $O(tn)$ counterclockwise/clockwise farthest enclosing sublist queries: Given a point $p\in P$, find from a set $\calL$ of sublists a sublist containing $p$ with the farthest counterclockwise/clockwise endpoint from $p$. Computing indices $a_i^j$ and $b_i^j$ takes $O(\log^2 n)$ time by Lemma~\ref{lem:firstout}. The counterclockwise/clockwise farthest enclosing sublist query problem has been studied in \cite{ref:TkachenkoDo25}; we summarize the result in the following lemma. 

\begin{lemma}{\em(\cite{ref:TkachenkoDo25})}
\label{lem:fes}
Given a set $\calL$ of $m$ sublists of $P$, we can construct a data structure for $\calL$ in $O(m\log m)$ time such that each counterclockwise/clockwise farthest enclosing sublist query can be answered in $O(\log m)$ time.
\end{lemma}

In our algorithm, the total size of the involved sets $\calL$ for which we need to construct the above data structure is $O(tn)$ in each $t$-th iteration. Hence, the total time of each $t$-th iteration is $O(tn\log^2 n)$. As there are $k$ iterations, the runtime of the whole algorithm is $O(k^2n\log^2 n)$. We conclude this section with the following theorem.

\begin{theorem}\label{thm:domsetunwgt}
Given a set of $n$ disks whose centers are in convex position in the plane, a smallest dominating set of the disk graph can be computed in $O(k^2n\log^2 n)$ time, where $k$ is the smallest dominating set size. 
\end{theorem}



\bibliographystyle{plainurl}
\bibliography{refs}
\end{document}